\documentclass[conference,letterpaper]{IEEEtran}

\usepackage[utf8]{inputenc} 
\usepackage[T1]{fontenc}
\usepackage{url}
\usepackage{ifthen}
\usepackage{cite,amssymb}
\usepackage[cmex10]{amsmath} 
\usepackage[dvipsnames]{xcolor}

\newtheorem{theorem}{Theorem}

\newtheorem{corollary}{Corollary}
\newtheorem{lemma}{Lemma}
\newtheorem{definition}{Definition}

\newtheorem{remark}{Remark}
\newenvironment{proof}[1][Proof]{\noindent\textbf{#1.} }{\ \rule{0.5em}{0.5em}}

\newcommand{\Ibec}{I_{\mathrm{BEC}}^{(\eta_t)}}
\newcommand{\Ibsc}{I_{\mathrm{BSC}}^{(t)}}
\newcommand{\Ibms}{I_{\mathrm{BMS}}^{(t)}}

\newcommand{\eps}{\varepsilon}

\begin{document}
\title{Lower Bounds on Mutual Information for Linear Codes Transmitted over Binary 
Input Channels,
and for Information Combining} 

\author{%
  \IEEEauthorblockN{Uri~Erez}
  \IEEEauthorblockA{Electrical Engineering\\
                    Tel Aviv University\\
                    Tel Aviv, Israel\\
                    Email: uri@eng.tau.ac.il}
    \and
      \IEEEauthorblockN{Or~Ordentlich}
  \IEEEauthorblockA{Computer Science and Engineering\\
                    Hebrew University of Jerusalem\\
                    Jerusalem, Israel\\
                    Email: or.ordentlich@mail.huji.ac.il}
  \and
  \IEEEauthorblockN{Shlomo Shamai (Shitz)}
  \IEEEauthorblockA{Electrical and Computer Engineering\\ 
                    Technion\\
                    Haifa, Israel\\
                    Email: sshlomo@ee.technion.ac.il}
}


\maketitle


\begin{abstract}
  It has been known for a long time that the mutual information between the input sequence and output of a binary symmetric channel (BSC) is upper bounded by the mutual information between the same input sequence and the output of a binary erasure channel (BEC) with the same capacity. Recently, Samorodintsky discovered that one may also lower bound the BSC mutual information in terms of the mutual information between the same input sequence and a more capable BEC. In this paper, we strengthen Samordnitsky's bound for the special case where the input to the channel is distributed uniformly over a linear code. Furthermore, for a general (not necessarily binary) input distribution $P_X$ and channel $W_{Y|X}$, we derive a new lower bound on the mutual information $I(X;Y^n)$ for $n$ transmissions of $X\sim P_X$ through the channel $W_{Y|X}$. 
\end{abstract}

\section{Introduction}

Let $P=P_{Y|X}$ and $Q=Q_{Z|X}$ be two channels with a common input alphabet $\mathcal{X}$ and output alphabets $\mathcal{Y},\mathcal{Z}$, respectively. Three common criteria for comparing/partially-ordering them are~\cite{CsiszarKornerBook,el2011network,PW_book}:
\begin{itemize}
	\item We say that  $P$ is degraded with respect to $Q$ if there exists a third channel $W=W_{Y|Z}$ with input alphabet $\mathcal{Z}$ and output alphabet $\mathcal{Y}$ such that $P=W\circ Q$, that is, $P$ is the composition of $W$ and $Q$.
	\item We say that $P$ is more noisy than $Q$ if for any distribution $P_{UX}$ we have that $I(U;Y)\leq I(U;Z)$, where $Y$ is obtained by feeding $X$ to $P$, and $Z$ is obtained by feeding $X$ to $Q$.
	\item We say that $P$ is less capable than $Q$ if for any distribution $P_{X}$ we have that $I(X;Y)\leq I(X;Z)$, where $Y$ is obtained by feeding $X$ to $P$, and $Z$ is obtained by feeding $X$ to $Q$.
\end{itemize}
Clearly if $P$ is degraded with respect $Q$, it is also more noisy than it, and similarly, if $P$ is more noisy than $Q$, it is also less capable than it. Furthermore, we have the following tensorization property for the three criteria: If $P$ is degraded with respect to (respectively, more noisy, less capable than) $Q$, then $P^{\otimes n}$ is degraded with respect to (respectively, more noisy, less capable than) $Q^{\otimes n}$~\cite[Problem 6.18]{CsiszarKornerBook},~\cite{PW17_SDPI,makur2018comparison}. Here, $P^{\otimes n}$ denotes the product channel from $\mathcal{X}^n$ to $\mathcal{Y}^n$ obtained by applying $P$ independently on each coordinate of $X^n$.

The above criteria are useful whenever computing  mutual information expressions involving $P$ is hard, whereas computing the same expressions with a channel $Q$ which dominates $P$ under the criterion relevant to the problem, is significantly easier. A canonical choice for $Q$ is the erasure channel, which outputs $Z=X$ with probability $1-e$ and outputs $Z=?$ with probability $e$. This is a convenient choice because computing mutual information expressions involving the erasure channel is often a feasible task, and furthermore, finding the ``dirtiest'' erasure channel (that is, the erasure channel with the largest $e$) that is less noisy than $P$ is equivalent to computing the strong data processing inequality (SDPI) coefficient of the channel $P$~\cite{PW17_SDPI}. The SDPI coefficient of a channel $P_{Y|X}$ is defined as
\begin{align*}
\eta(P_{Y|X})=\sup_{P_{UX}}\frac{I(U;Y)}{I(U;X)},
\end{align*}
where the supremum is with respect to all Markov triplets $U-X\stackrel{P_{Y|X}}{-}Y$. 

Computation of $\eta(P_{Y|X})$ reduces to computation of the SDPI coefficients of all binary sub-channels induced by $P_{Y|X}$, i.e., all channels obtained by restricting the input to two symbols $\{x_0,x_1\}\subset\mathcal{X}$~\cite{OP21}, for which closed-form expressions and bounds exist~\cite{PW17_SDPI}.

A special case that has received considerable attention in the literature is taking the channel $P_{Y|X}$ as a binary symmetric channel (BSC) with capacity $t\in[0,1]$. Let 
\begin{align*}
    h(p)=-p\log p-(1-p)\log(1-p),
\end{align*}
and let $h^{-1}$ be its inverse restricted to $[0,1/2]$, where throughout the paper all logarithms are taken to base $2$.
It is well known~\cite[Example 5.4]{el2011network} that this channel is degraded with respect to a binary erasure channel (BEC) with capacity $1-2h^{-1}(1-t)$, is more noisy than a BEC channel with capacity $(1-2h^{-1}(1-t))^2$, and is less capable than a BEC with capacity $t$. 
In particular, for any input distribution $X^n$ and $0\leq t\leq t_1\leq 1$, we have that 
$$\Ibsc(X^n;Y^n)\leq I_{\mathrm{BEC}}^{(t_1)}(X^n;Y^n),$$
where $\Ibsc(X^n;Y^n)$ denotes the mutual information between $X^n$ and the output of a memoryless BSC channel with capacity $t$, and $I_{\mathrm{BEC}}^{(t_1)}(X^n;Y^n)$ the mutual information between $X^n$ and the output of a memoryless BEC channel with capacity $t_1$. 

Thus, in cases where the mutual information between  the input vector and the output of a BEC channel can be computed/estimated, we immediately obtain upper bounds on the mutual information for the case of a BSC channel. 

In fact, it is well-known that among all binary-input memoryless output-symmetric (BMS) channels (see Definition~\ref{def:BMS} below) with the same capacity, the BEC is the most capable, and the BSC is the least capable~\cite{sasoglu2011thesis}. This implies that for any input distribution $X^n$, any BMS channel, and $0\leq t_{-}\leq t\leq t_{+}\leq 1$ we have
$$I_{\mathrm{BSC}}^{(t_-)}(X^n;Y^n)\leq \Ibms(X^n;Y^n)\leq I_{\mathrm{BEC}}^{(t_+)}(X^n;Y^n).$$
Thus, obtaining bounds in the other direction, i.e., lower bounds on $\Ibsc(X^n;Y^n)$ in terms of  $I_{\mathrm{BEC}}^{(t_0)}(X^n;Y^n)$ is desirable, as it enables to bound $\Ibms(X^n;Y^n)$ from above and below using only mutual information expressions involving the BEC.

However, deriving such bounds is a more challenging task. The main reason for this is that for $t\in(0,1)$, there is no $0<t_0$ for which the BEC with capacity $t_0$ is less capable than the BSC with capacity $t$. 

While lower bounds of the form $\Ibsc(X^n;Y^n)\geq I_{\mathrm{BEC}}^{(t_0)}(X^n;Y^n)$ that hold for any $X^n$ are impossible to obtain, in~\cite{Sam16} Samorodnitskty had the somewhat counter-intuitive observation that we can nevertheless lower bound $\Ibsc(X^n;Y^n)$ using the mutual information between $X^n$ and the output of a \emph{less} noisy BEC. In particular, he has shown that for any $X^n$ and any $t_1\geq \eta_t= (1-2h^{-1}(1-t))^2$, it holds that
\begin{align}
 \Ibsc(X^n;Y^n)\geq n\cdot \psi_{t}\left(\frac{I_{\mathrm{BEC}}^{(t_1)}(X^n;Y^n)}{nt_1} \right),   \label{eq:SamIneq}
\end{align}
where $\psi_{t}:[0,1]\to[0,t]$ is some increasing strictly convex function, to be explicitly specified later.\footnote{In fact,~\cite{Sam16} establishes this only for $t_1=\eta_t$, but by Lemma~\ref{lem:BECratiomonotone} proved in the appendix, this holds for all $t_1\geq \eta_t$.}

One of the most interesting applications of Samorodnitsky's result is for the case where $X^n$ is uniformly distributed on some linear code. In this case, the result implies that a code  will attain high mutual information when transmitted over the BSC channel, if it attains high mutual information when transmitted over a BEC channel (though with a different capacity). Our main result in this paper is an improvement of Samorodnitsky's result for this special case where $X^n\sim\mathrm{Unif}(C)$ and $C\subset\{0,1\}^n$ is a linear code. For this special case, we improve Samorodnitsky's bound from~\eqref{eq:SamIneq} to
\begin{align*}
   \Ibsc(X^n;Y^n)\geq  n\cdot \bar{\psi}_{t}\left(\frac{I_{\mathrm{BEC}}^{(t_1)}(X^n;Y^n)}{nt_1} \right), 
\end{align*}
where $\bar{\psi}_{t}(x):[0,1]\to [0,t]=t\cdot x$ is the upper concave envelope of  $\psi_{t}(x)$.

An important special case of a linear code is the repetition code. A uniform distribution on this code corresponds to $X^n=(X,\cdots,X)$ where $X\sim\mathrm{Bern}(1/2)$. For this case the problem of comparing $\Ibsc(X^n;Y^n)=\Ibsc(X;Y^n)$ to $I_{\mathrm{BEC}}^{(t_1)}(X^n;Y^n)=I_{\mathrm{BEC}}^{(t_1)}(X;Y^n)$ is referred to as the \emph{information combining} problem, which has been studied extensively in the literature~\cite{land2006information,sutskover2005extremes}. While the case of $X\sim\mathrm{Bern}(1/2)$ transmitted through $n$ copies of a BSC channel is handled by our main result, we further derive a lower bound for the general case where $X\sim P_X$ is transmitted $n$ times through a channel $W^{\otimes n}$, and show that
\begin{align}
I(X;Y^n)\geq \frac{I(P_X;W)}{\eta(P_X,W)}(1-(1-\eta(P_X,W))^n),
\label{eq:icboundintro}
\end{align}
where $Y^n$ is the output of the channel when $X$ is transmitted $n$ times, $I(P_X,W)=I(X;Y_1)$, and
\begin{align}
\eta(P_X,W)=\sup_{P_{U|X}}\frac{I(U;Y)}{I(U;X)}
\end{align}
is the input-dependent SDPI coefficient of the channel $W$ with input $P_X$. Note that we can further lower bound~\eqref{eq:icboundintro} as
\begin{align}
I(X;Y^n)\geq\frac{1-e^{-n\cdot\eta(P_X,W)}}{\eta(P_X,W)}\cdot I(P_X,W),
\end{align}
which is close to the obvious upper bound $n I(P_X;W)$ for $n\cdot \eta(P_X,W)\ll 1$. Thus, our bound essentially shows that when $n\cdot \eta(P_X,W)\ll 1$ each measurement  contributes about $I(P_X,W)$ bits of information to $I(X;Y^n)$ (as is the case for i.i.d. transmission).

\section{Main Result}

For a random vector $X^n$ on $\{0,1\}^n$, we denote by $I_{\mathrm{BEC}}^{(t)}(X^n;Y^n)$, respectively $\Ibsc(X^n;Y^n)$, the mutual information between $X^n$ and the output of a memoryless BEC channel, respectively BSC channel, with capacity $t$. We denote the strong data processing inequality (SDPI) coefficient of a BSC channel with capacity $t$ by~\cite{ahlswede1976spreading}
\begin{align}
\eta_t= (1-2h^{-1}(1-t))^2 .  
\end{align}
We further denote the ratio between the capacity and the SDPI coefficient by
\begin{align}
\alpha_t=\frac{t}{\eta_t}=\frac{t}{(1-2h^{-1}(1-t))^2}.    
\label{eq:alpha_t}
\end{align}
It can be verified that for all $0\leq t\leq 1$ we have $t \leq \eta_t$, and consequently, $\alpha_t\leq 1$. Furthermore, for all $0<t\leq 1,$
\begin{align*}
\alpha_t>\frac{\log_2(e)}{2}.    
\end{align*}
Our main result is the following.
\begin{theorem}
Let $C\subset\{0,1\}^n$ be a linear code, $u\in\{0,1\}^n$ be some shift, and $X^n=X^n_{C,u}\sim\mathrm{Uniform}(C+u)$. Then
\begin{align}
 \Ibsc(X^n;Y^n)\geq t\cdot \frac{\Ibec(X^n;Y^n)}{\eta_t}=\alpha_t\cdot\Ibec(X^n;Y^n).
 \label{eq:mainbound}
\end{align}
\label{thm:main}
\end{theorem}

\begin{remark}
We may rewrite~\eqref{eq:mainbound} as
\begin{align}
\frac{\Ibsc(X^n;Y^n)}{nt}\geq \frac{\Ibec(X^n;Y^n)}{n\eta_t},	
\end{align} 
indicating that for (shifted) linear codes, the fraction of capacity over the BSC with capacity $t$ is at least as large as the fraction of capacity over a BEC with capacity $\eta_t$.
\end{remark}

\begin{remark}
In~\cite[Theorem 12]{Sam16} (see also~\cite{PW17_SDPI} and~\cite{ordentlich16}), Samorodnitsky proved that for the BSC with capacity $t$, for any input $X^n$ on $\{0,1\}^n$ it holds that
\begin{align}
H(Y^n)\geq n \cdot\varphi_t\left(\frac{\Ibec(X^n;Y^n)}{n\cdot \eta_t} \right),  \label{eq:SamMGL} 
\end{align}
where $\varphi_t(x)=h\left(h^{-1}(1-t)\star h^{-1}(x)  \right)$ is the function from Mrs. Gerber's Lemma~\cite{wyner1973theorem}. Here,  $a\star b=a(1-b)+b(1-a)$ is the convolution between two numbers $a,b\in[0,1]$. Subtracting $n(1-t)$ from both sides of~\eqref{eq:SamMGL}, we obtain
\begin{align}
&\Ibsc(X^n;Y^n)=H(Y^n)-H(Y^n|X^n)=H(Y^n)-n(1-t)\nonumber\\
&\geq n \cdot\varphi_t\left(\frac{\Ibec(X^n;Y^n)}{n\cdot \eta_t} \right)-n(1-t)\nonumber\\
&=n\cdot \psi_{t}\left(\frac{\Ibec(X^n;Y^n)}{n\cdot \eta_t} \right),
\end{align}
where $\psi_{t}(x)=\varphi_t\left(x \right)-(1-t)$ is defined for $0\leq x\leq 1$. Since $x\mapsto\varphi_t(x)$ is convex, so is $x\mapsto \psi_{t}(x)$. Noting further that $\psi_{t}(0)=0$ and $\psi_{t}(1)= t$, convexity implies that $\psi_{t}(x)\leq t\cdot x$. In particular,
\begin{align}
\psi_{t}\left(\frac{\Ibec(X^n;Y^n)}{n\cdot \eta_t} \right)\leq t\cdot \frac{\Ibec(X^n;Y^n)}{n\cdot \eta_t}.
\end{align}
Comparing this with Theorem~\ref{thm:main}, we see that our bound is always at least as good as Samorodnitsky's. However, while Samorodnitsky's lower bound on $\Ibsc(X^n;Y^n)$ is valid for \emph{any} input $X^n$, our bound is only valid for $X^n$ uniform on a shifted linear code. In fact, it is easy to verify that Theorem~\ref{thm:main} does not hold if one does not impose any assumptions on $X^n$. Indeed, by the convexity of $t\mapsto g(t)=h(p\star h^{-1}(1-t))-(1-t)$, and the fact that $g(0)=0$ and $g(1)=h(p)$, it follows that $g(t)\leq t h(p)$. Thus, for $X^n\sim \mathrm{Bern}^{\otimes n}(p)$ we have
\begin{align}
\frac{\Ibsc(X^n;Y^n)}{t}=\frac{ng(t)}{t}\leq n h(p)=\frac{\Ibec(X^n;Y^n)}{\eta_t}.    
\end{align}
\end{remark}

\begin{definition}[BMS channels]
	A memoryless channel with binary input $X$ and output $Y$ is called \emph{binary-input memoryless output-symmetric (BMS)} if there exists a sufficient statistic $T(Y)=(X\oplus Z_A,A)$ for $X$, where $(A,Z_A)$ are statistically independent of $X$, and $Z_A$ is a binary random variable with $\Pr(Z_A=1|A=a)=a$.
 \label{def:BMS}
\end{definition}

It is well known and easy to verify that among all BMS channels with capacity $t$, the BEC is the most capable one, whereas the BSC is the least capable, see e.g.~\cite{sasoglu2011thesis}. Thus, the following is a straightforward corollary of Theorem~\ref{thm:main}.

\begin{corollary}
Under the assumptions of Theorem~\ref{thm:main}, for any BMS channel with capacity $t$ we have
\begin{align*}
\alpha_t\cdot\Ibec(X^n;Y^n)\leq \Ibms(X^n;Y^n)\leq I_{\mathrm{BEC}}^{(t)}(X^n;Y^n).
\end{align*}
\label{cor:BMS1}
\end{corollary}

Furthermore, since $t\leq \eta_t$, we have that the BEC with capacity $t$ is degraded with respect to the BEC with capacity $\eta_t$. Thus, the following statement immediately follows from Corollary~\ref{cor:BMS1}.
\begin{corollary}
Under the assumptions of Theorem~\ref{thm:main},
\begin{align*}
\alpha_t \cdot I_{\mathrm{BEC}}^{(t)}(X^n;Y^n)\leq  \Ibms(X^n;Y^n)\leq   I_{\mathrm{BEC}}^{(t)}(X^n;Y^n)
\end{align*}
where $\alpha_t$, defined in \eqref{eq:alpha_t}, satisfies $\alpha_t>\frac{\log_2(e)}{2}$, for  all \mbox{$0<t\leq 1$}.
\label{cor:BECtbound}
\end{corollary}

\begin{proof}[Proof of Theorem~\ref{thm:main}]
We may assume without loss of generality that $\mathrm{rank}(C)>0$, since otherwise $X^n$ is deterministic, so that $\Ibec(X^n;Y^n)=\Ibsc(X^n;Y^n)=0$ and the statement holds trivially.

Since $X^n$ is uniform over a shifted linear code with positive rank, we have that $X_i\sim\mathrm{Bern}(1/2)$ for all $i=1,\ldots,n$, and in particular 
\begin{align}
\Ibec(X_i;Y_i)=\eta_t,~\Ibsc(X_i;Y_i)=t,~~~\forall i=1,\ldots,n. 
\label{eq:capachievingmarginals}
\end{align}
This also implies the statement for $n=1$. We proceed by induction. Assume the statement holds for all linear codes and shifts in $\{0,1\}^{n-1}$.

Note that for any memoryless channel, and in particular for the BEC and the BSC, we have that
\begin{align}
&I(X^n;Y^n)=I(X^{n-1},X_n;Y^{n-1},Y_n)\nonumber\\
&=I(X^{n-1};Y^{n-1},Y_n)+I(X_n;Y^{n-1},Y_n|X^{n-1})\nonumber\\
&=I(X^{n-1};Y^{n-1})+I(X^{n-1};Y_n|Y^{n-1})+I(X_n;Y_n|X^{n-1})\nonumber\\
&=I(X^{n-1};Y^{n-1})+I(X_n;Y_n)-I(Y^{n-1};Y_n).
\label{eq:MIgenformula}
\end{align}
By~\eqref{eq:capachievingmarginals}, for the BEC, we therefore have that
\begin{align}
&\Ibec(X^n;Y^n)=\Ibec(X^{n-1};Y^{n-1})+\eta_t-\Ibec(Y^{n-1};Y_n)\nonumber\\
&=\Ibec(X^{n-1};Y^{n-1})+\eta_t-\eta_t\Ibec(Y^{n-1};X_n),
\label{eq:BECn}
\end{align}
while for the BSC, we have that
\begin{align}
&\Ibsc(X^n;Y^n)=\Ibsc(X^{n-1};Y^{n-1})+t-\Ibsc(Y^{n-1};Y_n)\nonumber\\
&\geq \Ibsc(X^{n-1};Y^{n-1})+t-\eta_t\Ibsc(Y^{n-1};X_n).\label{eq:BSCn}
\end{align}
In the last inequality we have used the strong data processing inequality (SDPI), stating that for any $U-X_n-Y_n$, where $P_{Y_n|X_n}$ is a BSC of capacity $t$,  we have that $I(U;Y_n)\leq \eta_t I(U;X_n)$. Since $Y^{n-1}-X_n-Y_n$ forms a Markov chain in this order, we can indeed apply the SDPI with $U=Y^{n-1}$ and obtain $\Ibsc(Y^{n-1};Y_n)\leq \eta_t \Ibsc(Y^{n-1};X_n)$.
We continue by noting that, since $X_n-X^{n-1}-Y^{n-1}$ forms a Markov chain in this order, we have
\begin{align}
\Ibsc&(Y^{n-1};X_n)\nonumber\\
&=\Ibsc(Y^{n-1};X^{n-1})-\Ibsc(Y^{n-1};X^{n-1}|X_n).  
\label{eq:MIdiff}
\end{align}
Substituting~\eqref{eq:MIdiff} into~\eqref{eq:BSCn}, gives
\begin{align}
\Ibsc(X^n;Y^n)&=(1-\eta_t)\Ibsc(X^{n-1};Y^{n-1})+t\nonumber\\
&+\eta_t\Ibsc(Y^{n-1};X^{n-1}|X_n).\label{eq:BSCn2}
\end{align}
The random variable $X^{n-1}$ is uniformly distributed over the projection of $C+u$ to the first $n-1$ coordinates. Since this projection is a shifted linear code in $\{0,1\}^{n-1}$, by the induction hypothesis, we have
\begin{align}
\Ibsc(X^{n-1};Y^{n-1})\geq \alpha_t \Ibec(X^{n-1};Y^{n-1}).
\label{eq:induction1}
\end{align}
Furthermore, conditioned on $X_n=0$ or $X_n=1$, we also have that $X^{n-1}$ is uniformly distributed over a shifted linear code in $\{0,1\}^{n-1}$ (though those shifted linear codes may differ for $X_n=0$ and $X_n=1$). Thus, again by the induction hypothesis
\begin{align}
&\Ibsc(X^{n-1};Y^{n-1}|X_n)\nonumber\\
&=\frac{1}{2}\Ibsc(X^{n-1};Y^{n-1}|X_n=0)\nonumber\\
&+\frac{1}{2}\Ibsc(X^{n-1};Y^{n-1}|X_n=1)\nonumber\\
&\geq \frac{\alpha_t}{2}\Ibec(X^{n-1};Y^{n-1}|X_n=0)\nonumber\\
&+\frac{\alpha_t}{2}\Ibec(X^{n-1};Y^{n-1}|X_n=1)\nonumber\\
&=\alpha_t \Ibec(X^{n-1};Y^{n-1}|X_n)\nonumber\\
&=\alpha_t\left(\Ibec(X^{n-1};Y^{n-1})- \Ibec(X_n;Y^{n-1})\right)
\label{eq:induction2}
\end{align}
where the last equality holds since $Y^{n-1}-X^{n-1}-X_n$ forms a Markov chain in this order, as in~\eqref{eq:MIdiff}. Substituting~\eqref{eq:induction1} and~\eqref{eq:induction2} into~\eqref{eq:BSCn2}, we obtain
\begin{align}
&\Ibsc(X^n;Y^n)\geq \alpha_t(1-\eta_t)\Ibec(X^{n-1};Y^{n-1})+t\nonumber\\
&+\alpha_t\eta_t\left(\Ibec(X^{n-1};Y^{n-1})- \Ibec(X_n;Y^{n-1})\right)\nonumber\\
&=\alpha_t\Ibec(X^{n-1};Y^{n-1})+t-\alpha_t\eta_t\Ibec(X_n;Y^{n-1})\nonumber\\
&=\alpha_t\left[\Ibec(X^{n-1};Y^{n-1})+\eta_t-\eta_t \Ibec(X_n;Y^{n-1})\right]\nonumber\\
&=\alpha_t \Ibec(X^n;Y^n),\label{eq:BSCn3}
\end{align}
where in the last equality we have used~\eqref{eq:BECn} and the fact that $\alpha_t=\frac{t}{\eta_t}$. This completes the proof.
\end{proof}


\section{Information Combining}
 
Let $X\sim P_X$, and let $W=W_{Y|X}$ be some channel with input alphabet $\mathcal{X}$ and output alphabet $\mathcal{Y}$. Assume $X$ is transmitted $n$ times through  $W$,  and the output is $Y^n=(Y_1,\ldots,Y_n)$. What can we say about $I(X;Y^n)$? Since the channel from $X^n=(X,\ldots,X)$ to $Y^n$ is memoryless, we have that
\begin{align}
I(X;Y^n)\leq\sum_{i=1}^n I(X;Y_i)=n I(X;Y)=n I(P_X,W).
\end{align} 
Combining this with the trivial upper bound $I(X;Y^n)\leq H(X)=H(P_X)$, we have that $I(X;Y^n)\leq \min\{H(P_X),nI(P_X,W)\}$. Denote by $\mathrm{EC}_e$ the erasure channel with input $\mathcal{X}$ whose output is $X$ with probability $1-e$ and $?$ with probability $e$. Let
\begin{align*}
C_{\mathrm{MC}}&=C_{\mathrm{MC}}(W)\nonumber\\
&=\min\{1-e\in[0,1]~:~W~\text{is less capable than } \mathrm{EC}_{e}\}.
\end{align*}
By tensorization of the more capable partial order, we have 
\begin{align}
I(X;Y^n)=I(P_X,W^{\otimes n})&\leq I(P,\mathrm{EC}_{1-C_{\mathrm{MC}}}^{\otimes n})\nonumber\\
&=(1-(1-C_{\mathrm{MC}})^n)H(P_X).
\label{eq:UBcombiningMC}
\end{align}
Our main result is a lower bound on $I(X;Y^n)$ taking a similar form to~\eqref{eq:UBcombiningMC} .
\begin{theorem}
Let $X\sim P_X$ and $W=W_{Y|X}$ be an input distribution and a channel with input-dependant SDPI coefficient satisfying $\eta(P_X,W)\leq \eta$.  Assume $X$ is transmitted $n$ times through  $W$,  and the output is $Y^n=(Y_1,\ldots,Y_n)$. Then,
\begin{align}
I(X;Y^n)\geq \alpha (1-(1-\eta)^n)H(P_X),
\label{eq:combining_lower_bound}
\end{align}
where 
\begin{align}
\alpha=\frac{I(P_X,W)}{\eta H(P_X)}.
\end{align}
\label{thm:combining}
\end{theorem}

\begin{proof}	
For $n=1$ the claim holds with equality. We proceed by induction. 
Starting from~\eqref{eq:MIgenformula}, we have
\begin{align}
	&I(X;Y^n)=I(X;Y^{n-1})+I(P_X,W)-I(Y^{n-1};Y_n)\nonumber\\
	&\geq I(X;Y^{n-1})+I(P_X,W)-\eta I(Y^{n-1};X)\label{eq:SDPIcombining}\\
 &=(1-\eta)I(X;Y^{n-1})+I(P_X,W),
	\label{eq:BSCnCombining}
\end{align}
where~\eqref{eq:SDPIcombining} follows from the strong data processing inequality, as $Y^{n-1}-X-Y_n$ forms a Markov chain in this order.
Using the induction hypothesis $I(X;Y^{n-1})\geq \alpha (1-(1-\eta)^{n-1})H(P_X)$, we further lower bound~\eqref{eq:BSCnCombining} as
\begin{align}
I(X&;Y^n)\geq \alpha  (1-(1-\eta)^{n-1})(1-\eta)H(P_X)+  I(P_X,W)\nonumber\\
&=\alpha\left[(1-(1-\eta)^{n-1})(1-\eta)H(P_X)+\frac{I(P_X;W)}{\alpha} \right]\nonumber\\
&=\alpha\left[(1-(1-\eta)^{n-1})(1-\eta)H(P_X)+\eta H(P_X) \right]\nonumber\\
&=\alpha(1-(1-\eta)^{n})H(P_X),
\end{align}
which establishes the claim.
\end{proof}

\begin{remark}
Recall that the information bottleneck curve corresponding to $(X,Y)\sim P_{XY}$ is defined as
\begin{align}
\mathrm{IB}_{P_{XY}}(R)=\max\left\{I(U;Y)~:~I(U;X)\leq R, U-X-Y \right\}.    \nonumber
\end{align}
Since $R\mapsto \mathrm{IB}_{P_{XY}}(R)\in[0,H(P_X)]$ is concave~\cite{tishby2000information,witsenhausen1975conditional}, and satisfies $\mathrm{IB}_{P_{XY}}(0)=0$ and $\mathrm{IB}_{P_{XY}}(H(P_X))=I(P_X;W)$, we have that $\mathrm{IB}_{P_{XY}}(R)\geq \frac{I(P_X;W)}{H(P_X)}R$. This implies that\footnote{In fact, the supremum in~\eqref{eq:SDPIasMaxRatio} is attained for $R\to 0$~\cite{anantharam2013maximal,makur2019information}}
\begin{align}
\eta(P_X,W)=\sup_{R\in(0,H(P_X)]}\frac{\mathrm{IB}_{P_{XY}}(R)}{R}\geq  \frac{I(P_X;W)}{H(P_X)}. 
\label{eq:SDPIasMaxRatio}
\end{align}
This, in turn, shows that $\alpha\leq 1$.
\end{remark}

\begin{remark}[Special case of $P_X=\mathrm{Bern}(1/2)$ and $W=\mathrm{BSC}$]
Note that for the special case of $\mathcal{X}=\{0,1\}$, $P_X=\mathrm{Bern}(1/2)$, and $W$ taken as a BSC with capacity $t$ the conclusion of Theorem~\ref{thm:combining} follows from Theorem~\ref{thm:main}. This follows since in this case $X^n=(X,\cdots,X)$ can be viewed as a random codeword in the repetition code (which is of course linear), and furthermore, for this choice of $W$ and $P_X$ we have~\cite{ahlswede1976spreading} that $\eta(P_X,W)=\eta(W)=(1-2h^{-1}(t))^2$. We may further relax Theorem~\ref{thm:combining}, as in Corollary~\ref{cor:BECtbound}, lower bounding $\eta_t=(1-2 h^{-1}(1-t))^2$ by $t$, to obtain that for $X\sim\mathrm{Bern}(1/2)$ we have
\begin{align}
 \alpha_t\cdot I_{\mathrm{BEC}}^{(t)}(X;Y^n) \leq \Ibsc(X;Y^n)\leq I_{\mathrm{BEC}}^{(t)}(X;Y^n).
\end{align}
Recalling that $\alpha_t\geq \frac{\log(e)}{2}$ for all $0<t\leq 1$, we obtain the uniform bound
\begin{align}
 \Ibsc(X;Y^n)\geq \frac{\log(e)}{2}I_{\mathrm{BEC}}^{(t)}(X;Y^n)=\frac{\log(e)}{2}(1-(1-t)^n).
 \label{eq:combiningBSCuniversal}
\end{align}
Equation~\eqref{eq:combiningBSCuniversal} may be tightened. Specifically, numerical evidence suggests that $I(X;Y^n)>0.92 (1-(1-t)^n)$,  which is significantly tighter than the uniform lower bound~\eqref{eq:combiningBSCuniversal} (recall that $\frac{\log(e)}{2}\approx 0.72$). To pursue such improvements, one may tighten the inequality in~\eqref{eq:SDPIcombining} by bounding $I(Y^{n-1};Y_n)\leq \mathrm{IB}_{P_{XY}}(I(Y^{n-1};X)$ which has a closed-form solution~\cite{erkip1998efficiency} for $P_X=\mathrm{Bern}(1/2)$ and $W=\mathrm{BSC}$.
\end{remark}

~~~~~~~
\section{Applications}
\begin{definition}
	We say that a code $C\subset\{0,1\}^n$ of rate $R$ is $\eps$-information-capacity achieving for the BEC, if for any $t>R$ we have that $I_{\mathrm{BEC}}^{(t)}(X^n;Y^n)\geq n(R-\eps)$, where $X^n\sim\mathrm{Uniform}(C)$.    
\end{definition}
We show that for codes that are $\eps$-information-capacity achieving for the BEC, the corresponding mutual information over the BSC cannot be too small. Results in similar spirit have been obtained in~\cite{hkazla2021codes}, see also~\cite{pathegama2023smoothing}.

\begin{theorem}
	Let $C\subset\{0,1\}^n$ be a linear code of rate $R$, that is $\eps$-information-capacity achieving for the BEC. Then,
	\begin{align}
		\Ibsc(X^n;Y^n)\geq n\left(1-\frac{\eps}{R} \right)\cdot\begin{cases}
			t & t<1-h\left(\frac{1-\sqrt{R}}{2} \right)\\
			\frac{t R}{\eta_t} & t\geq 1-h\left(\frac{1-\sqrt{R}}{2} \right)
		\end{cases}.
	\end{align}
\end{theorem}

\begin{proof}
	Let $t^*=1-h\left(\frac{1-\sqrt{R}}{2} \right)$, be such that $\eta_{t^*}=R$. Then,
	\begin{align}
		\frac{I_{\mathrm{BEC}}^{(\eta_{t^*})}(X^n;Y^n)}{\eta_{t^*}}\geq 
		\frac{n(R-\eps)}{R}=n\left(1-\frac{\eps}{R} \right).
	\end{align}
	By Lemma~\ref{lem:BECratiomonotone} in the Appendix,
	\begin{align}
		t\mapsto \frac{I_{\mathrm{BEC}}^{(t)}(X^n;Y^n)}{t}    
	\end{align}
	is non increasing, and hence, for all $t\leq t^*$, it holds that $\frac{I_{\mathrm{BEC}}^{(\eta_t)}(X^n;Y^n)}{\eta_t}\geq n\left(1-\frac{\eps}{R} \right)$. By Theorem~\ref{thm:main}, we therefore have that
	\begin{align}
		\Ibsc(X^n;Y^n)\geq nt\left(1-\frac{\eps}{R} \right),~~t\leq t^*.
	\end{align}
	For $t>t^*$ we have that $I_{\mathrm{BEC}}^{(t)}(X^n;Y^n)>nR(1-\eps)$, by the assumption that $C$ is a rate $R$ code that is $\eps$-information-capacity achieving for the BEC. Thus, by Theorem~\ref{thm:main}
 \begin{align}
		\Ibsc(X^n;Y^n)\geq t\cdot \frac{nR\left(1-\eps \right)}{\eta_t},~~t> t^*,
	\end{align}
 which establishes our claim.
\end{proof}

\appendix

\begin{lemma}
Fix an input distribution $X^n$ on $\{0,1\}^n$. The mapping $t\mapsto \frac{I_{\mathrm{BEC}}^{(t)}(X^n;Y^n)}{t} $ is non increasing.   
\label{lem:BECratiomonotone}
\end{lemma}

\begin{proof}
Let $0\leq t_1\leq t\leq 1$. Note that a BEC with capacity $t_1$ can be obtained by concatenating a BEC with capacity $t$, denoted $P^{\otimes n}_{Y|X}$ and a $\mathrm{EC}_{1-t'}$, where $t'=\frac{t_1}{t}$, denoted $W^{\otimes n}_{Z|Y}$, where $\mathcal{X}=\{0,1\}$ and $\mathcal{Y}=\mathcal{Z}=\{0,?,1\}$. Let $S_t\subset\{0,1\}^n$ denote the (random) set of indices not erased by $P^{\otimes n}_{Y|X}$ and $S_{t'}$ the (random) set of indices not erased by $W^{\otimes n}_{Y|X}$. For a set $S\subset\{0,1\}^n$ we denote by $X_S$ the restriction of $X^n$ to the indices included in $S$. We have
\begin{align}
I(X^n;Z^n)&=I(X^n;X_{S_t\cap S_{t'}},S_t \cap S_{t'})\nonumber\\
&=I(X^n;X_{S_t\cap S_{t'}}|S_t\cap S_{t'})\label{eq:SindepX}\\
&=H(X_{S_t\cap S_{t'}}|S_t\cap S_{t'})\nonumber\\
&=H(X_{S_t\cap S_{t'}}|S_t, S_{t'}),
\end{align}
where~\eqref{eq:SindepX} follows since $(S_t,S_{t'})$ are statistically independent of $X^n$.
We have
\begin{align}
H(X_{S_t\cap S_{t'}}|S_t, S_{t'})&=\mathbb{E}_{s_t\sim P_{S_t}}[H(X_{S_t'\cap s_t}|S_{t'},S_t=s_t)]\nonumber\\
    &\geq t'\cdot \mathbb{E}_{s_t\sim P_{S_t}}[H(X_{s_t}|S_t=s_t)]\label{eq:shearerLemma}\\
&=t'\cdot H(X_{S_t}|S_t)\nonumber\\
&=t'\cdot I(X^n;Y^n),
\end{align}
where in~\eqref{eq:shearerLemma} we have used Shearer's Lemma, see e.g.~\cite[Theorem 1.8]{PW_book}. Recalling that $t'=\frac{t_1}{t}$, we have therefore obtained that
\begin{align}
\frac{I(X^n;Z^n)}{t_1}\geq \frac{I(X^n;Y^n)}{t},    
\end{align}
as claimed.
\end{proof}

\section*{Acknowledgment}

The work of Or Ordentlich was supported by the ISRAEL SCIENCE 
FOUNDATION (ISF),
grant No.1641/21. 
The work of Uri Erez was supported by the ISRAEL SCIENCE 
FOUNDATION (ISF),
grant No. 588/19 and 736/23. The work of Shlomo Shamai (Shitz) was supported by the ISRAEL SCIENCE 
FOUNDATION (ISF),
grant No. 1897/19.
 \IEEEtriggeratref{11}

\newpage
\bibliographystyle{IEEEtran}
\bibliography{LessCapable}









\end{document}